\documentclass[12pt, 
]{article}
\usepackage{booktabs}

\usepackage[utf8]{inputenc}
\usepackage{amssymb}
\usepackage{amsmath}

\usepackage{geometry}
\newtheorem{theorem}{Theorem}

\newtheorem{definition}{Definition}
\newtheorem{example}{Example}
\newtheorem{lemma}{Lemma}
\newtheorem{proposition}{Proposition}
\newtheorem{remark}{Remark}

\newenvironment{proof}[1][Proof]{\emph{#1.} }{\  \hfill $\square $ \vspace{5 pt}}

\usepackage{natbib}
\usepackage{amssymb}
\usepackage[dvips]{graphics}
\usepackage{amsmath}

\usepackage{mathpazo}
\usepackage{enumerate}

\usepackage{amsfonts}

\usepackage{amssymb}

\usepackage{enumerate}

\usepackage{mathrsfs}

\usepackage[usenames]{color}
\usepackage{cancel}

\usepackage{pgf,tikz}

\usetikzlibrary{decorations.pathreplacing,angles,quotes}
\usetikzlibrary{arrows.meta}
\usetikzlibrary{arrows, decorations.markings}
\tikzset{myptr/.style={decoration={markings,mark=at position 1 with %
       {\arrow[scale=2,>=stealth]{>}}},postaction={decorate}}}

\usepackage{hyperref}
\hypersetup{
    colorlinks,
    citecolor=blue,
    filecolor=black,
    linkcolor=blue,
    urlcolor=blue
}

\usepackage{pgf,tikz}
\usetikzlibrary{arrows}

\usepackage{fancyhdr}

\usepackage{soul}

\usepackage{emptypage}

\usepackage[T1]{fontenc}
\DeclareFontFamily{T1}{calligra}{}
\DeclareFontShape{T1}{calligra}{m}{n}{<->s*[1.44]callig15}{}
\DeclareMathAlphabet\mathcalligra   {T1}{calligra} {m} {n}

\usepackage{multirow}
\usepackage{array}
\usepackage{mathtools}
\usepackage{arydshln}
\usepackage{threeparttable}
\usepackage{booktabs,caption}

\usepackage{ifthen}
\newboolean{showcomments}
\setboolean{showcomments}{true}

\newcommand{\pablo}[1]{  \ifthenelse{\boolean{showcomments}}
{\textcolor{green!50!black}{(T: #1)}}{}}
\newcommand{\marcelo}[1]{\ifthenelse{\boolean{showcomments}}
{\textcolor{red}{(M: #1)}}{}}
\newcommand{\agustin}[1]{  \ifthenelse{\boolean{showcomments}}
{\textcolor{blue!50!black}{(T: #1)}}{}}

\usepackage{orcidlink}

\usepackage{authblk}

\newcommand{\M}{\mathcal{O}}

\begin{document}

\title{\vspace{-2cm}On voting rules satisfying false-name-proofness and participation%
\thanks{%
We are grateful to Ulle Endriss, Jordi Mass\'o, William Thomson, Fernando Tohm\'e, Philippe Solal, and participants of SING20 in Maastricht and COMSOC 2025 in Vienna, for comments and suggestions that led to improvements of the paper.

Agust\'in Bonifacio acknowledges the financial support from UNSL through grants 031620 and 031323,
from Consejo Nacional de Investigaciones Científicas y Técnicas (CONICET) through grant PIP 112-200801-00655, and from Agencia Nacional de Promoción Científica y Tecnológica through grant PICT 2017-2355.
Federico Fioravanti acknowledges the financial support by the French National Research Agency within the project ANR-24-EXMA-0001 PEPR MathsVivES CONDORCET. Emails: \href{mailto:agustin.german.bonifacio@univ-st-etienne.fr}{agustin.german.bonifacio@univ-st-etienne.fr} and \href{mailto:federico.fioravanti@univ-st-etienne.fr}{federico.fioravanti@univ-st-etienne.fr} (corresponding author). 
}}


\author[1,2]{Agustín G. Bonifacio  
\orcidlink{0000-0003-2239-8673}}
\author[1]{Federico Fioravanti \orcidlink{0000-0002-7154-1192}}
\affil[1]{Universit\'e Jean Monnet Saint-\'Etienne, CNRS, Universit\'e Lyon 2, emlyon business school, GATE, 42023, Saint-\'Etienne, France.}
\affil[2]{Instituto de Matem\'{a}tica Aplicada San Luis, 
Universidad Nacional de San Luis and CONICET, San Luis, Argentina.}

\date{
}

\maketitle
\vspace{-1cm}
\begin{abstract}

We consider voting rules in settings where voters' identities are difficult to verify.
Voters can manipulate the process by casting multiple votes under different identities or abstaining from voting. 
Immunities to such manipulations are called \emph{false-name-proofness} and \emph{participation}, respectively. For the universal domain of (strict) preferences, these properties together imply \emph{anonymity} and are incompatible with \emph{neutrality}. For the domain of preferences defined over all subsets of a given set of objects, both \emph{false-name-proofness} and \emph{participation} cannot be met by rules that are also \emph{onto},  \emph{object neutral}, and  \emph{tops-only}.
However, when preferences over subsets of objects are restricted to be separable, all these properties can be satisfied. 
Furthermore, the domain of separable preferences is maximal for these properties.

\bigskip

\noindent \emph{JEL classification:} D71

\bigskip

\noindent \emph{Keywords:} false-name-proofness; participation; voting; tops-only; anonymity; neutrality.  

\end{abstract}

\section{Introduction}

Societies make decisions by means of voting rules, mapping profiles of voters' preferences into social alternatives. 
In highly anonymous settings, such as the Internet, there are various ways a voter can manipulate the voting rule. 
When participants' identities cannot be easily verified, or when the number of participants is unknown, opportunities for manipulation arise. One such manipulation involves a voter using multiple identities to cast several votes. We say that a rule immune to voters casting duplicate votes is ``false-name-proof''. 
More generally, a voting rule is ``strongly false-name-proof'' if a voter does not benefit from submitting multiple (and possibly different) votes.\footnote{Our strong false-name-proofness property is typically called false-name-proofness in the literature \citep[see, for example,][]{yokoo2004effect,conitzer2008anonymity,bu2013unfolding}.} 
In this paper, we focus on false-name-proof manipulation, which we view as the cognitively simplest form of strategic behavior. 
For instance, one may imagine a voter sitting at a computer repeatedly refreshing a webpage and selecting the same option each time, without engaging in any substantive reasoning about alternative strategies that might benefit her.
A voter can also benefit by abstaining from voting, leading to what is known in the literature as the no-show paradox \citep{fishburn1983paradoxes,MOULIN198853}.
We say that a rule that does not allow such behavior induces ``participation''. 
Since defining these properties requires a changing set of active voters, 
 we consider societies with a variable set of voters.

We are interested in studying voting rules that satisfy false-name-proofness and induce  participation 
in two different settings. In the first, alternatives do not have any specific structure. In the second,  alternatives consist of subsets of a given set of objects (candidates, binary issues, or alike).  

When social alternatives are unstructured and all preferences over those alternatives are admissible, i.e., when we consider the universal domain of preferences,  results on voting rules satisfying some form of false-name-proofness are rather negative. \citet{bu2013unfolding} shows that strong false-name-proofness implies both ``strategy-proofness'' (no voter ever gains by untruthful voting) and ``anonymity'' (exchanging voters' identities does not affect the choice made by the rule).  
As is well-known from the result of \citet{gibbard1973manipulation} and \citet{satterthwaite1975strategy}, for three or more alternatives, there are no non-constant strategy-proof and anonymous rules defined in the universal domain.  
Therefore, no non-constant strongly false-name-proof rule defined in the universal domain exists either. 
However, the weakening of strong false-name-proofness may allow for some possibility results.\footnote{The property of strategy-proofness is central to the literature on mechanism design and has been extensively studied \citep[see][for a comprehensive survey]{barbera2011strategyproof}. We deliberately depart from this line of inquiry and instead focus on rules that satisfy false-name-proofness and participation.}

Our first goal is to analyze the existence of voting rules that satisfy false-name-proofness and participation on the universal domain of preferences. In most voting settings, it is common to also assume that no alternative  deserves special consideration. The requirement of ``neutrality'' formalizes this by demanding that the exchanging of alternatives' names does not affect the choice made by the rule.
Although in most of the literature concerning false-name-proofness anonymity is implicitly assumed, we demonstrate that voting rules satisfying false-name-proofness and participation are anonymous (Proposition~\ref{partfnpimplyanon}).
As a consequence, we also show that there are no neutral rules that  satisfy our two requirements of immunity to manipulation (Proposition~\ref{nonanonandneutral}).

When the set of  alternatives consists of all the subsets of a given set of objects, an important restricted domain of preferences is that of ``separable'' preferences: adding an object to a set leads to a better set if and only if the object is ``good'' (as a singleton set, the object is better than the empty set). 
On that restricted domain, \cite{fioravanti2024false} characterize all voting rules that are false-name-proof, strategy-proof, and ``onto''  (every subset of objects is a possible outcome) as the subclass of voting by quota rules \citep{barbera1991voting} where, to be chosen,  each object needs either at least one vote or a unanimous vote.

Our second goal is to analyze what happens when separability is relaxed, i.e. when all preferences over subsets of objects are admissible.\footnote{A typical example of when this domain can be deemed relevant is inspired by \citet{barbera1991voting}. 
Suppose you are on a university professor hiring committee. You might think that Borda and Condorcet are outstanding professors, and would love to have any of them employed, but believe that the department will be chaos with the two of them in it (probably because of some dissidence on how they like to vote).} 
When alternatives are sets, assuming separable preferences rules out complementarities and substitution effects across objects, a restriction not implied by the voting problem itself. 
The unrestricted domain therefore provides a natural benchmark for evaluating the robustness of voting rules.
Besides false-name-proofness and participation, we impose three other desirable properties. First is ontoness. As we previously said, it means that no subset of objects should be discarded from consideration \emph{a priori}. Second, the structure of this restricted domain allows us to define the weaker neutrality axiom of ``object neutrality'', by which exchanging  objects' names does not affect the choice made by the rule. Third, as voters may not be willing to submit full preferences (this seems particularly important, for example, in online voting settings), we also require the informational simplicity property of ``tops-onliness'', by which only the top choices of the voters are relevant to the rule. We demonstrate that no rule satisfies false-name-proofness, participation, as well as  these three additional desiderata (Theorem \ref{mainresult}). 
Even though it might seem to be over-demanding to require voting rules to satisfy so many properties, the impossibility result is far from being straightforward since (i) we show that the five axioms are independent on the domain of all preferences over subsets of objects, and (ii) the rules characterized by \citet{fioravanti2024false} satisfy all these axioms on the domain of separable preferences.\footnote{\citet{fioravanti2024false} characterize the voting by quota 1 and voting by unanimous quota as the unique object-neutral rules which satisfy strategy-proofness, false-name-proofness, and ontoness, in the domain of separable preferences. 
It is not difficult to see that these rules also satisfy all the aforementioned axioms.}

Finally, we ask whether and to what extent the domain of separable preferences can be enlarged while maintaining the compatibility of all five properties. It turns out that  such a domain is maximal for those properties: adding any non-separable preference to it entails losing at least one of the properties involved (Theorem \ref{theo maximal}).

The property of (strong) false-name-proofness was introduced by \cite{yokoo2004effect}, for the problem of assigning objects with transfers to agents with quasi-linear preferences. 
In voting environments,  when  preferences are single-peaked \citep{black1948rationale}, \citet{todo2011false} characterize the class of all strongly false-name-proof, anonymous, and ``efficient'' (no voter can be made better off without  some voter being made worse off) voting rules. 
\cite{todo2011false} and \cite{todo2020false} extend the analysis to the case where the preferences respect a tree structure. 
Moreover, \citet{Nehama2022} works with a combinatorial structure, generalizing many previous results.
\cite{conitzer2008anonymity} characterizes the class of  anonymous and neutral probabilistic voting rules over a finite set of alternatives that satisfy strong false-name-proofness and participation. The members of this class are parameterized by a probability $p\in[0,1]$, and they are defined as follows. 
With probability $p$, an alternative is chosen uniformly at random.  With probability $1-p$, a pair of alternatives is chosen uniformly at random.  
If all voters  prefer one alternative to the other in the pair, the preferred alternative is chosen; otherwise, a fair coin is used to decide between the two. 
Although \citeauthor{conitzer2008anonymity}'s  \citeyearpar{conitzer2008anonymity} result implies the non-existence of neutral and deterministic voting rules that satisfy strong false-name-proofness and participation, our Proposition~\ref{nonanonandneutral} is not a direct corollary of his, as we use false-name-proofness, a weaker axiom.

Strong false-name-proofness and participation have also been studied in environments where voting is costly.
\citet{wagman2008optimal} characterize, in the case of two alternatives, all voting rules that are most responsive to votes and satisfy participation, together with either strong false-name-proofness or group strong false-name-proofness.
They also establish analogous results for three alternatives and propose computational methods for settings with four or more alternatives.
Such models are particularly relevant in blockchain governance, where voting protocols are designed to deter duplicate identities by making participation costly, while simultaneously encouraging participation through explicit rewards for voting \citep{scaroundtheblockGrossi2022,blockchaingovernanceKiayias2023}.

The plan of the paper is as follows.
Section \ref{model} presents the basic notions and axioms that we use, while we present the results in Section \ref{results}.
Finally, Section \ref{conclusion} contains some concluding remarks.

\section{Model}\label{model}

Let $\mathcal{N}$ be the family of all finite and non-empty subsets 
of the set of positive integers $\mathbb{Z}_+$. 
An element $N\in\mathcal{N}$ is interpreted as a society. 
We denote the cardinality of $N$ by $n$ and refer to an element $i\in N$ as a \emph{voter}. 
Each set of voters $N\in \mathcal{N}$ has to collectively choose 
an \emph{alternative} from a set $\mathcal{A}$. Let $\mathscr{U}_\mathcal{A}$ denote the set of all strict linear orders over $\mathcal{A}$. Each voter $i$ is endowed with a \emph{preference} $P_i \in \mathscr{U}_\mathcal{A}$, where $A\ P_i\ B$ means that for voter $i$, alternative $A$ is preferred to alternative $B$. 
We denote the weak counterpart of $P_i$ by $R_i$. 
Thus, $A\mathrel{R_i} B$ implies that either $A\ P_i\ B$ or $A=B$.
When a preference order is not attached to a particular voter, we write it as $P_0$. For each $N \in \mathcal{N}$, a \emph{profile}  is an ordered list  of preferences $P_N=(P_i)_{i \in N} \in \mathscr{U}_\mathcal{A}^N$.  
Given a preference $P_i\in \mathscr{U}_\mathcal{A}$,  denote by  $t(P_i)\in \mathcal{A}$  the \emph{top} alternative for voter $i$ and denote by  $b(P_i)\in \mathcal{A}$ to the \emph{bottom} alternative for voter $i$.

We call $\mathscr{U}_\mathcal{A}$ the \emph{universal domain} of preferences. 
Besides studying this domain, we will be interested in the domain arising from considering as alternatives all the subsets of a given set of \emph{objects} $\mathcal{O}=\{1,\ldots, O\}$ with $O\geq 2$, i.e., the case where $\mathcal{A}=2^\mathcal{O}$. 
Abusing notation, we call $\mathscr{U}_\mathcal{O}$ the \emph{domain (of preferences) over subsets of objects}. 
Notice that, after some renaming of the alternatives involved, any domain over subsets of objects can be considered a universal domain but, in general, there are universal domains that cannot be considered as domains over subsets of objects.\footnote{For example, if $|\mathcal{A}|=3$ then there is no set of objects $\mathcal{O}$ such that $\mathcal{A}=2^\mathcal{O}$.} 

Given a domain $\mathscr{D} \subseteq \mathscr{U}_{\mathcal{A}}$, let $\mathscr{D}^{\mathcal{N}}=\bigcup_{N \in \mathcal{N}} \mathscr{D}^N$.  A \emph{voting rule} on $\mathscr{D}$ is a mapping $f:\mathscr{D}^{\mathcal{N}} \longrightarrow \mathcal{A}$ that assigns, for each $N \in \mathcal{N}$ and each $P_N \in \mathscr{D}^N$, an element $f(P_N) \in \mathcal{A}$. Next, we define desirable properties of voting rules. 
To do this, fix a domain $\mathscr{D}$ and a rule $f:\mathscr{D}^{\mathcal{N}} \longrightarrow \mathcal{A}$.

The first property states that all alternatives could be selected in all societies.

\vspace{5 pt}

\noindent \textbf{Ontoness:} For each $N \in \mathcal{N}$ and each $A \in \mathcal{A}$, there is a profile $P_N\in \mathscr{D}^N$ such that $f(P_N)=A$.

\vspace{5 pt}

The next axiom asserts that all essential information for the voting rule is found in the top alternatives of the voters. 
Thus, the rule requires minimal information from the voters.

\vspace{5 pt}

\noindent \textbf{Tops-onliness:}
For each $N \in \mathcal{N}$ and each pair of profiles $P_N,P'_N\in \mathscr{D}^{N}$ such that $t(P_i)=t(P'_i)$ for all $i\in N$, it is the case that $f(P_N)=f(P'_N)$.

\vspace{5 pt}

The following three properties are particularly relevant in contexts such as online voting, where a social planner cannot easily verify voters' identities or determine the total number of participants. 
The first property asserts that no voter should ever gain by casting  repeated votes.

\vspace{5 pt}

\noindent \textbf{False-name-proofness:}
For each $N,N' \in \mathcal{N}$ with $N \cap N'=\emptyset$, each $i \in N$, each $P_{N} \in \mathscr{D}^{N}$, and each $P_{N'} \in \mathscr{D}^{N'}$ such that $P_j=P_i$ for each $j \in N'$, we have $f(P_N) \, R_i \, f(P_{N\cup N'}).$

\vspace{5 pt}

\cite{conitzer2008anonymity}'s related condition imposes stronger restrictions on the voting rule by not requiring that the additional preferences submitted by voter $i\in N$ coincide with voter $i$'s original preference $P_i$. 

\vspace{5 pt}

\noindent \textbf{Strong false-name-proofness:}
For each $N,N' \in \mathcal{N}$ with $N \cap N'=\emptyset$, each $i \in N$, each $P_{N} \in \mathscr{D}^{N}$, and each $P_{N'} \in \mathscr{D}^{N'}$, we have $f(P_N)\,R_{i}\,f(P_{N\cup N'}).$

\vspace{5 pt}

The next axiom states that voters are never harmed by voting. 

\vspace{5 pt}

\noindent \textbf{Participation:}
For each $N \in \mathcal{N}$ with $|N|\geq 2$, each $i \in N$, and each $P_{N} \in \mathscr{D}^{N}$, we have $f(P_N) \, R_i \, f(P_{N\setminus \{i\}}).$

\vspace{5 pt}

The following property states that no voter should receive a differential treatment.

\vspace{5 pt}

\noindent \textbf{Anonymity:} For each permutation $\sigma:\mathbb{Z}_+ \longrightarrow \mathbb{Z}_+$, each $N\in \mathcal{N}$, and each $P_N\in \mathscr{D}^N$, $f(\sigma(P_N))=f(P_N)$,  where $\sigma(P_N)=(P_{\sigma(i)})_{i\in N}$.

\vspace{5 pt}

\cite{conitzer2008anonymity} merges the properties of \emph{strong false-name-proofness, participation} and \emph{anonymity} under the name of \emph{anonymity-proofness}.
A  property similar to \emph{anonymity}, but applied to alternatives, is provided next. Given a permutation $\gamma:\mathcal{A}\longrightarrow \mathcal{A}$, and a profile $P_N\in \mathscr{D}^{N}$, let $P_N^\gamma$ be the profile defined by setting, for each $i \in N$ and each pair $A,A' \in \mathcal{A}$, $\gamma(A) \ P^\gamma_i \ \gamma(A')$ if and only if $A P_i   A'$.

\vspace{5 pt}

\noindent \textbf{Neutrality:}
For each permutation $\gamma:\mathcal{A}\longrightarrow \mathcal{A}$ and each $P_N\in \mathscr{D}^{N}$, $\gamma (f(P_N))=f(P_N^\gamma)$.

\vspace{5 pt}

A weaker notion of neutrality can be written  for voting rules defined on the domain of subsets of objects, $\mathscr{U}_\mathcal{O}$. Given a permutation $\mu:\mathcal{O}\longrightarrow \mathcal{O}$, a subset of objects $S\in 2^\mathcal{O}$, and a profile $P_N\in \mathscr{U}_\mathcal{O}^{N}$, let $\mu (S)=\{ \mu(x) : x \in S\}$ and let $P_N^\mu$ be the profile such that, for each $i \in N$ and each pair $S,T \in 2^\mathcal{O}$, $\mu(S)\ P^\mu_i\ \mu(T)$ if and only if $S\ P_i\ T$.

\vspace{5 pt}

\noindent \textbf{Object neutrality:} 
For each permutation $\mu:\mathcal{O}\longrightarrow \mathcal{O}$, each $N \in \mathcal{N}$, and each $P_N\in \mathscr{U}_\mathcal{O}^{N}$, $\mu (f(P_N))=f(P_N^\mu).$

\section{Results}\label{results}

\subsection{Universal domain}\label{subsection universal}

Our first result shows that if the identity of a voter changes while the preference remains the same, the outcome of the rule remains unchanged. This result is instrumental in proving one of our main findings: that any false-name-proof voting rule satisfying participation is anonymous.
Let $\mathscr{D}\subseteq \mathscr{U}_\mathcal{A}$, that is, a generic subset of the universal domain.

\begin{proposition}\label{singleanonymity}
Let $f:\mathscr{D}^{\mathcal{N}} \longrightarrow \mathcal{A}$ be a voting rule that satisfies \emph{false-name-proofness} and \emph{participation}. 
Let $N \in \mathcal{N}$, $i\in N$,  and $P_N \in \mathscr{D}^N$. 
Then, if $i^\star \notin N$ and $P_{i^\star} \in \mathscr{D}$ is such that $P_{i^\star}=P_i$, it is the case that $f( P_{(N\cup\{i^\star\})\setminus\{i\}})=f(P_N).$   
\end{proposition}

To show that Proposition~\ref{singleanonymity} holds, we use the following result, which helps us get rid of repeated votes.
\begin{lemma}\label{repeatedvotes}
Let $f:\mathscr{D}^{\mathcal{N}} \longrightarrow \mathcal{A}$ be a voting rule that satisfies \emph{false-name-proofness} and \emph{participation}. 
Let $N \in \mathcal{N}$, $i\in N$, and $P_N \in \mathscr{D}^N$; and let   
$i^\star \notin N$ and $P_{i^\star} \in \mathscr{D}$ be such that $P_{i^\star}=P_i$. Then, 
$f(P_N)=f(P_{N\cup\{i^\star\}})$.  
\end{lemma}
\begin{proof}
Let $f$, $N$, $P_N$, $i$, $i^\star$, and $P_{i^\star}$ be as in the lemma. By \emph{false-name-proofness}, $f(P_N)\,R_i\, f(P_{N\cup\{i^\star\}})$.
By \emph{participation},  $f(P_{N\cup\{i^\star\}})\,R_{i^\star}\,f(P_N)$.
As $P_{i^\star}=P_i$, we have that $f(P_{N\cup\{i^\star\}})\,R_{i}\,f(P_N)$.
Thus, $f(P_N)\ R_i\, f(P_{N\cup\{i^\star\}})\,R_{i}\,f(P_N)$ and  $f(P_N)=f(P_{N\cup\{i^\star\}})$.
\end{proof}

\bigskip 

\noindent \begin{proof}[Proof of Proposition~\ref{singleanonymity}]
Let $f$, $N$, $P_N$, $i$, $i^\star$, and $P_{i^\star}$ be as in the proposition. By Lemma~\ref{repeatedvotes} and \emph{participation}, $f(P_N)=f(P_{N\cup\{i^\star\}})\,R_i\,f(P_{N\setminus\{i\}\cup\{i^\star\}})$. By Lemma~\ref{repeatedvotes} and \emph{participation} again, $f(P_{N\setminus\{i\}\cup\{i^\star\}})=f(P_{N\cup\{i^\star\}})\,R_{i^\star}\,f(P_N)$.
As $P_i=P_{i^\star}$, $f(P_N)\,R_i\,f(P_{N\setminus\{i\}\cup\{i^\star\}})\,R_if(P_N)$  and, therefore, $f(P_N)=f(P_{N\setminus\{i\}\cup\{i^\star\}})$.
\end{proof}

\citet{bu2013unfolding} and \citet{fioravanti2024false} explore the connection between \emph{false-name-proofness} and \emph{anonymity}. Our next result follows that line and shows that the names of the voters are not important for a rule that satisfies \emph{false-name-proofness} and \emph{participation}.

\begin{proposition}\label{partfnpimplyanon}
A voting rule $f:\mathscr{D}^{\mathcal{N}} \longrightarrow \mathcal{A}$ that satisfies \emph{false-name-proofness} and \emph{participation}, also satisfies \emph{anonymity}.   
\end{proposition} 
\begin{proof}
Let $f$ satisfy \emph{false-name-proofness} and \emph{participation}. 
Consider $N \in \mathcal{N}$, a profile \mbox{$P_N\in \mathscr{D}^N$}, and a permutation $\sigma:\mathbb{Z}_+ \longrightarrow \mathbb{Z}_+$. 
We need to show that $f(\sigma(P_N))=f(P_N)$,  where $\sigma(P_N)=(P_{\sigma(i)})_{i\in N}$.  
There are two cases to consider:
\begin{itemize}
    \item[$\boldsymbol{1}.$] $\boldsymbol{N \cap \sigma(N)=\emptyset}$. 
    By iterating the result of Proposition \ref{singleanonymity}, we obtain that $f(\sigma(P_N))=f(P_N)$.
    \item[$\boldsymbol{2}.$] $\boldsymbol{N \cap \sigma(N) \neq \emptyset}$. 
    Let $N'=\sigma(N)$ and  $N''\in\mathcal{N}$ be such that $N''\cap (N\cup N')=\emptyset$ and $|N''|=|N|$. Then, there are two permutations $\widetilde{\sigma}, \widehat{\sigma}: \mathbb{Z}_+ \longrightarrow \mathbb{Z}_+$
such that  $\widetilde{\sigma}(N)=N''$, $\widehat{\sigma}(N'')=N'$, and $\sigma=\widehat{\sigma} \circ \widetilde{\sigma}$.  
By the previous case, $f(\sigma(P_N))=f(P_{N'})=f(\widehat{\sigma}(P_{N''}))=f(P_{N''})=f(\widetilde{\sigma}(P_N))=f(P_N)$.  
\end{itemize} Therefore, $f$ satisfies \emph{anonymity}.

 \end{proof}

\begin{remark} \em
    Both requirements in Proposition \ref{partfnpimplyanon} are necessary to obtain \emph{anonymity}, i.e., there are non-anonymous voting rules that satisfy either {\em false-name-proofness} or {\em participation}. 
The rule 
that selects voter $1$'s top alternative whenever 1 is present, and otherwise assigns a status-quo alternative, satisfies \emph{participation} and is not \emph{anonymous}.  
 Furthermore, the rule 
 that selects voter 1's bottom alternative whenever 1 is present, and otherwise assigns a status-quo alternative, satisfies \emph{false-name-proofness} and is not \emph{anonymous}.\footnote{A natural setting in which such rules may be relevant is a platform like X (formerly Twitter). 
 Suppose a social planner wishes to follow the views of a particular public figure; if that individual does not maintain an account, however, another user could create an impostor profile and exploit the resulting influence.} 
\end{remark}

It is well known that, for most choices of $|N|$ and $|\mathcal{A}|$, 
\emph{anonymity} and \emph{neutrality} are not compatible on the universal domain.\footnote{In fact, such compatibility exists if and only if $|\mathcal{A}|$ cannot be written as the sum of dividers of $|N|$ different than 1 \citep[see, for example, Exercise 9.9 in][for more details]{moulin1991axioms}.} 
This implies that both requirements cannot be met for voting rules defined in a variable population environment.

\begin{remark}\label{remark anon and neutral}
\em No voting rule $f:\mathscr{U}_\mathcal{A}^{\mathcal{N}} \longrightarrow \mathcal{A}$ satisfies \emph{anonimity} and \emph{neutrality}.  
\end{remark}

Our first impossibility result says that for rules defined in the universal domain, \emph{false-name-proofness} together with \emph{participation} are incompatible with \emph{neutrality}.

\begin{proposition}\label{nonanonandneutral}
 No voting rule $f:\mathscr{U}_\mathcal{A}^{\mathcal{N}} \longrightarrow \mathcal{A}$ satisfies \emph{false-name-proofness, participation}, and \emph{neutrality}.   
\end{proposition}
\begin{proof}
    Let $f$ satisfy \emph{false-name-proofness} and \emph{participation}. Then, by Proposition \ref{partfnpimplyanon}, $f$ satisfies \emph{anonymity}. By Remark \ref{remark anon and neutral}, $f$ cannot be \emph{neutral}. 
\end{proof}

\begin{remark} \em
    
    Theorem 1 of \citet{conitzer2008anonymity} shows that only a very restricted class of randomized rules can satisfy these properties, which in particular implies that no deterministic and neutral voting rule meets them. 
    In his framework, however, the weaker notion of false-name-proofness admits a much richer family of admissible rules; for instance, any rule that depends solely on the set of votes in a profile satisfies the axiom. 
    Consequently, Proposition \ref{nonanonandneutral} does not follow directly from \citet{conitzer2008anonymity}.
\end{remark}

\subsection{Domain over subsets of objects}\label{subsection subsets}

Now, we turn our attention to rules defined on the domain over subsets of objects, $\mathscr{U}_\mathcal{O}$. In this environment, it makes sense to relax \emph{neutrality} to \emph{object neutrality}, in order to look for positive results.
In the following example we show the existence of \emph{false-name-proof} rules that also satisfy \emph{participation} and \emph{object neutrality}. Alas, neither \emph{ontoness} nor \emph{tops-onliness} are guaranteed.

\begin{example}\label{ejemplo} \em
First, define rule $f^\mathcal{O}:\mathscr{U}_\mathcal{O}^\mathcal{N} \longrightarrow 2^\mathcal{O}$ as follows. 
For each $N \in \mathcal{N}$ and each $P_N \in \mathscr{U}_{\mathcal{O}}^N$, $f^\mathcal{O}(P_N)=\mathcal{O}$. 
This constant rule always selects the whole set of objects $\mathcal{O}$ and clearly satisfies all properties but \emph{ontoness}. 

Next, given $N \in \mathcal{N}$ and $P_N \in \mathscr{U}_{\mathcal{O}}^N$, let 
$\widetilde{\mathcal{O}}(P_N)=\{i \in N : t(P_i)\neq \mathcal{O}$ and $ \mathcal{O} \ P_i  \ S$ for each $S \in 2^\mathcal{O} \setminus \{t(P_i), \mathcal{O}\}\}$. 
Define rule $\widetilde{f}:\mathscr{U}_\mathcal{O}^\mathcal{N} \longrightarrow 2^\mathcal{O}$ as follows. For each $N \in \mathcal{N}$ and each $P_N \in \mathscr{U}_{\mathcal{O}}^N$,  
$$\widetilde{f}(P_N)=\begin{cases}
    t(P_i) & \text{ if } i \in \widetilde{\mathcal{O}}(P_N) \text{ and }t(P_i)=t(P_j) \text{ for each }j \in  \widetilde{\mathcal{O}}(P_N)\\
    \mathcal{O} & \text{ otherwise}
\end{cases}$$
Rule $\widetilde{f}$ selects the set of all objects, $\mathcal{O}$, unless all voters who consider  $\mathcal{O}$ as their second choice share their top choice, in which case the rule recommends such top choice. 
This rule satisfies all properties but \emph{tops-onliness}. 
To see this, let $\mathcal{O}=\{x,y\}$, $N \in \mathcal{N}$, and consider $P_N, P_N' \in \mathscr{U}_\mathcal{O}^N$ such that $t(P_i)=t(P_i')=\{x\}$ for each $i \in N$, $b(P_{i})=\mathcal{O}$ for each  $i \in N$, and $\widetilde{O}(P_N')=N.$ 
Then, $\widetilde{f}(P_N)=\mathcal{O}\neq \{x\}=\widetilde{f}(P_N')$. \hfill $\Diamond$
\end{example}

Our second impossibility result establishes that, on the domain of preferences over subsets of objects, the five desired properties are mutually incompatible. 
The intuition of the impossibility is that, for any given preference profile, \emph{ontoness}, \emph{tops-onliness}, \emph{false-name-proofness}, and \emph{participation} jointly pin down the outcome of the voting rule. 
Then, \emph{false-name-proofness} and \emph{participation}  help to determine outcomes in reduced societies consisting of one or two voters. 
When analyzing these reduced societies, we obtain contradictions between \emph{object-neutrality} and \emph{anonymity} (the latter implied by Proposition~\ref{partfnpimplyanon}).

\begin{theorem}\label{mainresult} No voting rule $f:\mathscr{U}_\mathcal{O}^{\mathcal{N}} \longrightarrow 2^\M$ satisfies \emph{ontoness, tops-onliness, false-name-proofness, participation}, and \emph{object neutrality}.
\end{theorem}
\begin{proof}
Assume there is a voting rule $f:\mathscr{U}_\mathcal{O}^{\mathcal{N}} \longrightarrow 2^\M$ that satisfies the five axioms. First, we claim that there are $N \in \mathcal{N}$, a profile $P_N \in \mathscr{U}_\mathcal{O}^N$, and a voter $i \in N$, such that 
\begin{equation}\label{eq neq sin i}
 f(P_N)\neq f(P_{N\setminus\{i\}}).   
\end{equation}
If this is not the case, for each $\{j,k\} \in \mathcal{N}$ and each  $(P_j, P_k) \in \mathscr{U}_\mathcal{O}^{\{j,k\}}$, we have  $f(P_j)=f(P_j,P_k)=f(P_k)$ and thus $f(P_j)=f(P_k)$, implying that all one-voter societies are assigned the same alternative. This violates \emph{ontoness}. So \eqref{eq neq sin i} holds.

Next, let $P_i' \in  \mathscr{U}_\mathcal{O}$ be such that $t(P_i')=t(P_i)$ and $t(P_i') \ R_i' \ f(P_{N\setminus \{i\}}) \ P_i' \  T$ for each $T \in 2^{\mathcal{O}}\setminus \{t(P_i), f(P_{N\setminus \{i\}})\}$. By \emph{participation}, $f(P_i',P_{N\setminus \{i\}}) \ R_i' \ f(P_{N\setminus \{i\}})$. By \emph{tops-only}, $f(P_i', P_{N \setminus \{i\}})=f(P_N)$. Thus, \eqref{eq neq sin i} implies $f(P_i', P_{N \setminus \{i\}})\neq f(P_{N\setminus\{i\}})$ and, therefore, $f(P_i', P_{N \setminus \{i\}})=t(P_i')$. Hence, using \emph{tops-only} again, we get 
\begin{equation}\label{f de P es top de i}
    f(P_N)=t(P_i).
\end{equation}
By Lemma~\ref{repeatedvotes}, we can safely assume that all the tops in $P_N$ are different. Let $j \in N\setminus \{i\}$ and consider $P_j' \in \mathscr{U}_\mathcal{O}$ such that $t(P_j')=t(P_j)$ and $b(P_j')=t(P_i)$.  
By \emph{tops-only} and \eqref{f de P es top de i}, $f(P_j', P_{N \setminus \{j\}})=t(P_i)$.  
By \emph{participation}, $f(P_j',P_{N \setminus \{j\}})\ R_j' \ f(P_{N\setminus \{j\}})$ and, since $b(P_j')=t(P_i)$, we have $f(P_{N \setminus \{j\}})=t(P_i)$. 
Removing in the same way each remaining voter, one at a time, we obtain $f(P_i)=t(P_i)$. 
Thus, by \emph{object neutrality}, $f(P_i^\mu)=\mu(f(P_i))$ for any permutation $\mu: \mathcal{O} \longrightarrow \mathcal{O}$. 
 Together with \emph{anonymity}, granted by Proposition~\ref{partfnpimplyanon}, this implies that 
\begin{equation}\label{ele prima}
    f(P_\ell')=t(P_\ell') \text{ for each } \{\ell\}\in \mathcal{N} \text{ and  } P_\ell' \in \mathscr{U}_\mathcal{O} \text{ with } |t(P_\ell')|=|t(P_i)|. 
\end{equation}

Now, let $\{j,k\} \in \mathcal{N}$ and  $(P_j,P_k) \in \mathscr{U}_\mathcal{O}^{\{j,k\}}$ be such that $t(P_j)\neq t(P_k)$ and $|t(P_j)|=|t(P_k)|=|t(P_i)|.$ There are three cases:
\begin{itemize}
    \item[$\boldsymbol{1.}$] $\boldsymbol{f(P_j, P_k) \notin \{t(P_j), t(P_k)\}}$. Let $P_j' \in \mathscr{U}_\mathcal{O}$ be such that $t(P_j')=t(P_j)$ and $b(P_j')=f(P_j,P_k)$. By \emph{tops-only}, $f(P_j',P_k)=f(P_j,P_k)$. By \emph{participation}, $f(P_j',P_k) \ R_j' \ f(P_k)$. Thus, $f(P_k)=f(P_j,P_k)\neq t(P_k)$, contradicting \eqref{ele prima}.
    \item[$\boldsymbol{2.}$] $\boldsymbol{f(P_j, P_k)=t(P_j)}$.  Consider a permutation $\mu: \mathcal{O} \longrightarrow \mathcal{O}$ such that $\mu(t(P_j))=t(P_k)$ and $\mu(t(P_k))=t(P_j)$.
By \emph{object neutrality} we obtain $f(P^\mu_{j}, P^\mu_{k})=\mu(f(P_j, P_k))=\mu(t(P_j))=t(P_k)$, and thus $f(P^\mu_{j}, P^\mu_{k})=t(P_k)$. 
By Proposition~\ref{partfnpimplyanon}, $f$ is \emph{anonymous}. 
Therefore, by \emph{anonymity} and \emph{tops-only}, $$f(P_j, P_k)=f(P_k,P_j)=f(P_j^\mu, P_k^\mu)=t(P_k),$$ and so $f(P_j,P_k)=t(P_k)$, contradicting this case's hypothesis.
    \item[$\boldsymbol{3.}$] $\boldsymbol{f(P_j, P_k)=t(P_k)}$.  
    A similar reasoning to the previous case allows us to deduce that $f(P_j,P_k)=t(P_j)$, contradicting this case's hypothesis.
\end{itemize}
Since in each case we reach a contradiction, we conclude that no such rule $f$ exists. 
\end{proof}

\bigskip

It is important to notice that 
there are no redundant axioms in Theorem \ref{mainresult}. To see this, 
we consider several rules. 
Each one satisfies all the axioms but one. 

\begin{itemize}
\item All but \emph{ontoness}: The rule $f^\mathcal{O}$ in Example~\ref{ejemplo}.
\item All but \emph{tops-onliness}: The rule $\widetilde{f}$ in Example \ref{ejemplo}.
    \item All but \emph{false-name-proofness}: For each $N \in \mathcal{N}$ and each $P_N \in \mathscr{U}_{\mathcal{O}}^N$, $f^{\min}(P_N)=t(P_{\min\{i \ : \  i \in N\}}).$
    This rule selects the top of the voter with the minimum index in the society.
    This rule satisfies \emph{participation} but is not \emph{anonymous}, thus, by Proposition~\ref{partfnpimplyanon}, $f^{\min}$ is not \emph{false-name-proof}. 
    \item All but \emph{participation}: For each $N \in \mathcal{N}$ and each $P_N \in \mathscr{U}_{\mathcal{O}}^N$, $x \in f^\star(P_N)$ if and only if \mbox{$|\{i\in N\mid x\in t(P_i)\}|=1$}.
    This rule selects those objects that belong to only one top-choice set of the preference profile. 
To see that $f^\star$ does not satisfy \emph{participation}, let $\mathcal{O}=\{x,y,z\}$,  $N=\{i,j\}$ and $(P_i, P_j) \in \mathscr{U}_{\mathcal{O}}^{\{i,j\}}$ be such that $t(P_i)=\{x,y\}$, $\{x,y,z\} \ P_i \ \{z\}$, and $t(P_j)=\{x,y,z\}$. 
Then, if voter $i$ does not participate, she can manipulate $f^\star$  since $f^\star(P_j)=\{x,y,z\} \ P_i \ \{z\}=f^\star(P_i,P_j)$.
\item All but \emph{object-neutrality}: Let $\succ$ be a linear order over $2^\mathcal{O}$.
For each $N \in \mathcal{N}$ and each $P_N \in \mathscr{U}_{\mathcal{O}}^N$, $f^{\succ}(P_N)=\max_{\succ}\{t(P_i) : i \in N\}$. 
This rule selects, for each profile, the best positioned top according to $\succ$. 
To see that $f^\succ$ does not satisfy \emph{object-neutrality}, let $\mathcal{O}=\{x,y,z\}$, $\{x\}\succ \{y\}\succ \{z\}$, $N=\{i,j\}$, and $P_N \in \mathscr{U}_\mathcal{O}^N$ be such that $t(P_i)=\{x\}$ and $t(P_j)=\{z\}$. 
Thus $f(P_N)=\{x\}$.
If we consider a permutation $\mu: \mathcal{O} \longrightarrow \mathcal{O}$ such that $\mu(x)=z,$  $\mu(y)=x,$ and $\mu(z)=y$, then $\mu(f^\succ(P_N))=\mu(\{x\})=\{z\} \neq \{y\}=f^\succ(P_N^\mu)$. 

\end{itemize}
\subsection{Domain of separable preferences: maximality}\label{subsection separable}
        We have seen that when we consider the domain of all preferences over subsets of objects, no voting rule satisfies \emph{ontoness, tops-onliness, false-name-proofness, participation}, and \emph{object neutrality.} 
        Nevertheless, rules exist that satisfy all of them when  preferences are separable.  
        We can find two examples in \citet{fioravanti2024false}, with voting by quota 1 and voting by unanimous quota. 
        A natural question is whether there is a domain larger than the domain of separable preferences in which 
        voting rules still satisfy all the axioms.
        Next, we answer the latter in the negative.

First, we remember the definition of separability.  For a voter, an object is \emph{good} by itself if it is better 
than no object at all; otherwise, the object is \emph{bad}. 
A preference  is separable if the partition of the object set into good and bad objects guides the ordering of subsets, in the sense that adding a good
object to any set leads to a better set while adding a bad object to any set leads to a worse set. Formally, 
preference $P_0$  
is \emph{%
separable} if for each $S\in 2^{\mathcal{O}}$ and each $x \in \mathcal{O} \setminus S$, 
\begin{equation*}
S\cup \{x \} \ P_0 \ S\text{ if and only if }\{x \} \ P_{0} \ \emptyset .
\end{equation*}
Let $\mathscr{S}$ be the domain of separable preferences. An important characterization of separability is presented in the following remark.
\begin{remark}{\normalfont{\citep{barbera1991voting}}}\label{remark separable} 
\em 
     Preference $P_0 \in \mathscr{U}_\mathcal{O}$ is separable if, for each $S\in 2^{\mathcal{O}}$ and each $x \in \mathcal{O} \setminus S$,  
     $$S\cup\{x\}\,P_0\ S \text{ if and only if } x\in t(P_0).$$
\end{remark}

Before introducing our definition of maximality, we need to define how to extend a rule from the domain of separable preferences to a broader domain.
\begin{definition}
Given a \emph{tops-only} rule $f:\mathscr{S}^\mathcal{N}\longrightarrow 2^\mathcal{O}$ and a domain $\mathscr{S}^\star$ such that $\mathscr{S} \subseteq \mathscr{S}^\star$, the \textbf{\emph{tops-only} extension of $\boldsymbol{f$ to $\mathscr{S}^\star}$} is such that, for each $N \in \mathcal{N}$ and each $P_N \in \mathscr{S}^\star$, $f(P_N)=f(\overline{P}_N)$ for some $\overline{P}_N \in \mathscr{S}$ with $t(\overline{P}_i)=t(P_i)$ for each $i \in N$.   
\end{definition}
In words, the output of the tops-extension rule on any given profile ``mimics'' the outcome that the rule has in a profile with the same tops that belongs to the domain of separable preferences.
This makes this extension unique.

        
Let $\mathcal{F}$ denote the class of all rules defined on the domain of separable preferences that are \emph{onto, tops-only, false-name-proof}, satisfy \emph{participation}, and are \emph{object neutral}.
The following definition, inspired by \citet{BONIFACIO2023102845} and \citet{arribillaga2025}, formalizes the idea of a maximal domain for a set of rules satisfying a list of properties.\footnote{Previous studies on maximal domains mostly focus on the property of \emph{strategy-proofness} \citep[see, for example,][]{serizawa1995power,ching1998maximal,masso2001maximal}.}




\begin{definition}\label{def maximal}
Let $\mathscr{S}^\star$ be such that  $\mathscr{S} \subseteq \mathscr{S}^\star\subseteq \mathscr{U}_\mathcal{O}$ and let $\mathcal{F}^\star \subseteq \mathcal{F}$. Domain $\mathscr{S}^\star$ is \textbf{maximal for} $\boldsymbol{\mathcal{F}^\star}$ if

    \begin{itemize}
        \item[(i)] for each $f\in\mathcal{F}^\star$ the \emph{tops-only} extension of $f$ to $\mathscr{S}^\star$ satisfies  \emph{ontoness, false-name-proofness},
        \emph{participation} and \emph{object neutrality}, and 

        \item[(ii)] for each $P_0 \in \mathscr{U}_{\mathcal{O}} \setminus  \mathscr{S}^\star$ there is  $f\in\mathcal{F}^\star$ such that the \emph{tops-only} extension of  $f$ to $\mathscr{S}^\star \cup \{P_0\}$ violates (at least) one of the properties listed in (i).
    \end{itemize}
\end{definition}
Note that we define the notion of maximality with respect to the domain of separable preferences, as this is our object of study, but it can be defined with respect to any domain.


An important fact about the maximality of a domain with respect to a list of properties thus defined is its monotonicity: the bigger the set of rules considered for maximality, the smaller the domain of preferences in which the properties hold. We highlight this observation in the following remark.

\begin{remark} \label{remark maximal}  \em Assume that $\mathscr{S}^\star$ is maximal for $\mathcal{F}^\star$. If $\mathcal{F}^i \subseteq \mathcal{F}^\star$ and $\mathscr{S}^i$ is maximal for $\mathcal{F}^i$ with $i \in \{1,2\}$, then by Definition \ref{def maximal} it follows that $\mathscr{S}^\star \subseteq \mathscr{S}^1 \cap \mathscr{S}^2$.
\end{remark}

Our maximality result establishes that the domain of separable preferences is maximal for the family of rules that are \emph{onto, tops-only, false-name-proof}, satisfy \emph{participation}, and are \emph{object neutral}. 
The proof proceeds by showing that, for any non-separable preference, one can construct a tops-only extension of a rule that satisfies all the properties on the separable domain but fails one of them, in this case \emph{participation}, for that non-separable preference.



\begin{theorem}\label{theo maximal}	
The domain of separable preferences is maximal for the set of all \emph{onto, tops-only}, and \emph{false-name-proof} rules that satisfy \emph{participation} and \emph{object neutrality}, i.e., $\mathscr{S}$ is maximal for $\mathcal{F}$.
\end{theorem}
\begin{proof}
First, we introduce the following notation: given a society $N\in\mathcal{N}$ and profile $P_N\in \mathscr{U}_{\mathcal{O}}^N$,  let $t(P_N)=\{t(P_i) : i\in N\}$ be the collection of (different) tops of profile $P_N$.
Now, let $\overline{P}_0 \in \mathscr{U}_\mathcal{O} \setminus \mathscr{S}$. By Remark \ref{remark separable}, there are $S \subseteq \mathcal{O}$ and $x \in \mathcal{O}\setminus S$ such that either
\begin{equation}\label{sep 1}
    x \in t(\overline{P}_0) \text{ \ and \ }S \ \overline{P}_0 \ S \cup \{x\}
\end{equation}
or 
\begin{equation}\label{sep 2}
    x \notin t(\overline{P}_0) \text{ \ and \ }S \cup \{x\} \ \overline{P}_0 \ S. 
\end{equation}

 Consider the rule $f^>:\mathscr{S}^\mathcal{N} \longrightarrow 2^\mathcal{O}$ such that, for each $N \in \mathcal{N}$ and  each profile $P_N\in \mathscr{S}^N$ satisfies $$x\in f^>(P_N)\text{ if and only if }\left|\left\{t(P_i)\in t(P_N) \ :  \ x\in t(P_i)\right\}\right|> \frac{|t(P_N)|}{2}.$$
Clearly, $f^>$ is \emph{tops-only} and \emph{object neutral}. 
Since the rule only depends on the set of top subsets of options and not on how many times each subset appears, $f^>$ is \emph{false-name-proof}.
Moreover, as casting a vote can only add support to a good object for a voter, $f^>$ also satisfies \emph{participation}. 
Therefore,  $f^> \in\mathcal{F}$. 

Assume further that $t(\overline{P}_0)\neq \emptyset$ and consider the \emph{tops-only} extension of $f^>$ to $\mathscr{S} \cup \{\overline{P}_0\}$. 
There are two cases to consider. 
If \eqref{sep 1} holds, let $(P_1, P_2) \in (\mathscr{S} \cup \{\overline{P}_0\})^{\{1,2\}}$ be such that $t(P_1)=S$ and $t(P_2)=S \cup \{x\}$. 
Let $i^\star \notin \{1,2\}$ and endow voter $i^\star$ with preference $\overline{P}_0$. 
Notice that, by \eqref{sep 1}, $t(\overline{P}_{i^\star})\neq S\cup\{x\}$.
Then, $$f^>(P_1, P_2)=S \ \overline{P}_{i^\star} \ S \cup \{x\}=f^>(P_1, P_2, \overline{P}_{i^\star}),$$ contradicting \emph{participation}. 
If \eqref{sep 2} holds and $t(\overline{P}_0)\nsubseteq S$, let $y \in t(\overline{P}_0) \setminus S$  and 
let $(P_1, P_2,P_3) \in (\mathscr{S} \cup \{\overline{P}_0\})^{\{1,2,3\}}$ be such that $t(P_1)=S$, $t(P_2)=S \cup \{x\}$, and $t(P_3)=S \cup \{x,y\}$. Let $i^\star \notin \{1,2,3\}$ and endow voter $i^\star$ with preference $\overline{P}_0$.
Then, $$f^>(P_1, P_2,P_3)=S \cup \{x\} \ \overline{P}_{i^\star} \ S =f^>(P_1, P_2,P_3, \overline{P}_{i^\star}),$$ contradicting \emph{participation}. 
If \eqref{sep 2} holds and $t(\overline{P}_0)\subseteq S$, let $y \in t(\overline{P}_0)$  and 
consider $(P_1, P_2,P_3) \in (\mathscr{S} \cup \{\overline{P}_0\})^{\{1,2,3\}}$ such that $t(P_1)=S$, $t(P_2)=S \cup \{x\}$, and $t(P_3)=(S \setminus \{y\} )\cup \{x\}$. 
Let $i^\star \notin \{1,2,3\}$ and endow voter $i^\star$ with preference $\overline{P}_0$.
Notice that, by \eqref{sep 2}, $t(\overline{P}_{i^\star})\neq S$.
Then, $$f^>(P_1, P_2,P_3)=S \cup \{x\} \ \overline{P}_{i^\star} \ S =f^>(P_1, P_2,P_3, \overline{P}_{i^\star}),$$ contradicting \emph{participation}.  
Since in both cases we reach a contradiction, it follows that $t(\overline{P}_0)=\emptyset$.  
This implies that 
    \begin{equation} \label{max s1}
        \text{if }\mathscr{S}^>\text{ is maximal for }\{f^>\},\text{ then }\mathscr{S}^> \subseteq \mathscr{S}\cup\{\overline{P}_0 \in \mathscr{U}_\mathcal{O} \setminus \mathscr{S} : t(\overline{P}_0)=\emptyset\}. 
    \end{equation}

Second, consider the rule 
 $f^\geq:\mathscr{S}^\mathcal{N} \longrightarrow 2^\mathcal{O}$ defined by setting, for each $N \in \mathcal{N}$ and each profile $P_N\in \mathscr{S}^N$, $$x\in f^\geq(P_N)\text{ if and only if }\left|\left\{t(P_i)\in t(P_N) \ :  \ x\in t(P_i)\right\}\right|\geq \frac{|t(P_N)|}{2}.$$
A reasoning similar to the one used for $f^>$ shows that $f^\geq \in\mathcal{F}$.

Assume now that $t(\overline{P}_0)\neq \mathcal{O}$ and consider the \emph{tops-only} extension of $f^\geq$ to $\mathscr{S} \cup \{\overline{P}_0\}$. 
There are two cases. 
If \eqref{sep 1} holds and $S \neq \emptyset$, let $(P_1, P_2, P_3) \in (\mathscr{S} \cup \{\overline{P}_0\})^{\{1,2,3\}}$ be such that $t(P_1)=S$, $t(P_2)=S \cup \{x\}$, and $t(P_3)=\emptyset$. Let $i^\star \notin \{1,2,3\}$ and endow voter $i^\star$ with preference $\overline{P}_0$.
Notice that, by \eqref{sep 1}, $t(\overline{P}_{i^\star})\neq S\cup\{x\}$.
Then, $$f^\geq(P_1, P_2,P_3)=S \ \overline{P}_{i^\star} \ S \cup \{x\}=f^\geq(P_1, P_2,P_3, \overline{P}_{i^\star}),$$ contradicting \emph{participation}. 
If \eqref{sep 1} holds and $S = \emptyset$, let $y \in \mathcal{O}\setminus t(\overline{P}_0)$ and let $(P_1, P_2, P_3) \in (\mathscr{S} \cup \{\overline{P}_0\})^{\{1,2,3\}}$ be such that $t(P_1)=\emptyset$, $t(P_2)=\{x\}$, and $t(P_3)=\{y\}$. Let $i^\star \notin \{1,2,3\}$ and endow voter $i^\star$ with preference $\overline{P}_0$.
Notice that, by \eqref{sep 1}, $t(\overline{P}_{i^\star})\neq \{x\}$.\footnote{Also note that, if $\mathcal{O}=\{x,y\}$, then \eqref{sep 1} is trivially contradicted.}
Then, $$f^\geq(P_1, P_2,P_3)=\emptyset \ \overline{P}_{i^\star} \ \{x\}=f^\geq(P_1, P_2,P_3, \overline{P}_{i^\star}),$$ contradicting \emph{participation}. 
If \eqref{sep 2} holds, let  $(P_1, P_2) \in (\mathscr{S} \cup \{\overline{P}_0\})^{\{1,2\}}$ be such that $t(P_1)=S$ and $t(P_2)=S \cup \{x\}$.
Let $i^\star \notin \{1,2\}$ and endow voter  $i^\star$ with preference $\overline{P}_0$. 
Notice that by \eqref{sep 2}, $t(\overline{P}_{i^\star})\neq S$.
Then, $$f^\geq(P_1, P_2)=S \cup \{x\} \ \overline{P}_{i^\star} \ S =f^\geq(P_1, P_2, \overline{P}_{i^\star}),$$ contradicting \emph{participation}.  Since in each case we reach a contradiction, it follows that $t(\overline{P}_0)=\mathcal{O}$.  
This implies that 
    \begin{equation} \label{max s2}
        \text{if }\mathscr{S}^\geq\text{ is maximal for }\left\{ f^\geq \right\},\text{ then }\mathscr{S}^\geq \subseteq \mathscr{S}\cup\{\overline{P}_0 \in \mathscr{U}_\mathcal{O} \setminus \mathscr{S} : t(\overline{P}_0)=\mathcal{O}\}. 
    \end{equation}

Finally, let $\mathscr{S}^\star$ be maximal for $\mathcal{F}$. 
By Definition \ref{def maximal}, $\mathscr{S} \subseteq \mathscr{S}^\star$. 
Since $f^>, f^\geq \in \mathcal{F}$, Remark \ref{remark maximal} implies that $\mathscr{S}^\star \subseteq \mathscr{S}^> \cap \mathscr{S}^\geq$. 
By \eqref{max s1} and \eqref{max s2}, $\mathscr{S}^> \cap \mathscr{S}^\geq \subseteq \mathscr{S}$ and thus $\mathscr{S}^\star \subseteq \mathscr{S}$. 
Hence, $\mathscr{S}^\star=\mathscr{S}.$
\end{proof}

It is important to emphasize that Theorem \ref{mainresult} does not follow from Theorem \ref{theo maximal}. 
The latter should be understood as showing that, within the class of rules satisfying all the properties on the separable domain, the introduction of even a single non-separable preference guarantees the existence of some rule in the class that violates participation. 
This observation, by itself, does not rule out the possibility that a particular rule might satisfy all axioms on the full, unrestricted domain. 
Rather, it underscores that once the family of admissible rules on the separable domain is fully characterized, the separable domain is maximal in a stronger, structural sense.

\begin{remark}\em
Given that the property violated in the proof of Theorem~\ref{theo maximal} is \emph{participation}, our definition of maximality implies that the domain of separable preferences is also maximal for the family of all \emph{tops-only} rules that satisfy this axiom.    
\end{remark}

\section{Conclusion}\label{conclusion}

We study voting rules that satisfy false-name-proofness and participation in environments where voters' identities are difficult to verify. 
Our analysis establishes three main results.

First, we show that false-name-proofness and participation jointly imply anonymity across any domain of preferences (Proposition~\ref{partfnpimplyanon}). 
This reveals a fundamental constraint: rules that prevent both duplicate voting and beneficial abstention cannot differentiate between voters. 
Consequently, no neutral voting rule satisfies both properties in the universal domain (Proposition~\ref{nonanonandneutral}), extending previous impossibility results to the weaker notion of false-name-proofness.

Second, when alternatives consist of subsets of objects, we prove that ontoness, tops-onliness, false-name-proofness, participation, and object neutrality cannot be simultaneously satisfied (Theorem~\ref{mainresult}). 
This impossibility is non-trivial: the five axioms are logically independent, and some voting by quota rules satisfy all of them on the domain of separable preferences.

Third, we establish that the domain of separable preferences is maximal for these five properties (Theorem~\ref{theo maximal}). 
For any non-separable preference, there exists at least one rule in the class whose tops-only extension violates at least one axiom. 
This demonstrates that separability cannot be relaxed without sacrificing compatibility for some rules satisfying all desiderata on the separable domain.

Our findings have practical implications for designing voting mechanisms in online environments where identity verification is costly or infeasible. 
They suggest that platforms seeking to implement manipulation-resistant voting with the additional desiderata of ontoness, tops-onliness, and object neutrality must either restrict the domain of admissible preferences or accept that some natural rules satisfying all properties on separable preferences will fail when preferences exhibit complementarities or substitution effects. 
The separable domain thus provides a well-defined benchmark for such mechanisms, where all desiderata can be jointly satisfied.

\section*{Statements and Declarations}
The authors have no competing interests to declare that are relevant to the content of this article.

\bibliographystyle{ecta}
\bibliography{ref2}

@inproceedings{blockchaingovernanceKiayias2023,
author = {Kiayias, Aggelos and Lazos, Philip},
title = {SoK: Blockchain Governance},
year = {2023},
isbn = {9781450398619},
publisher = {Association for Computing Machinery},
address = {New York, NY, USA},
url = {https://doi.org/10.1145/3558535.3559794},
doi = {10.1145/3558535.3559794},
booktitle = {Proceedings of the 4th ACM Conference on Advances in Financial Technologies},
pages = {61–73},
numpages = {13},
location = {Cambridge, MA, USA},
series = {AFT '22}
}

@inproceedings{scaroundtheblockGrossi2022,
author = {Grossi, Davide},
title = {Social Choice Around the Block: On the Computational Social Choice of Blockchain},
year = {2022},
isbn = {9781450392136},
publisher = {International Foundation for Autonomous Agents and Multiagent Systems},
address = {Richland, SC},
booktitle = {Proceedings of the 21st International Conference on Autonomous Agents and Multiagent Systems},
pages = {1788–1793},
numpages = {6},
location = {Virtual Event, New Zealand},
series = {AAMAS '22}
}

@article{arribillaga2025,
  title = {Not obviously manipulable allotment rules},
  author = {R. Pablo Arribillaga and Agustín G. Bonifacio},
  journal = {Economic Theory},
  year = {2025},
  volume = {80},
  number = {1},
  pages = {355--380},
  isbn = {1432-0479},
  doi = {10.1007/s00199-024-01633-1},
  url = {https://doi.org/10.1007/s00199-024-01633-1}
}

@Article{Nehama2022,
author={Nehama, Ilan
and Todo, Taiki
and Yokoo, Makoto},
title={Manipulation-resistant false-name-proof facility location mechanisms for complex graphs},
journal={Autonomous Agents and Multi-Agent Systems},
year={2022},
month={Jan},
day={22},
volume={36},
number={1},
pages={12},
issn={1573-7454},
doi={10.1007/s10458-021-09535-5},
url={https://doi.org/10.1007/s10458-021-09535-5}
}

@article{black1948rationale,
  title={On the rationale of group decision-making},
  author={Black, Duncan},
  journal={Journal of Political Economy},
  volume={56},
  number={1},
  pages={23--34},
  year={1948},
  publisher={The University of Chicago Press}
}

@article{gibbard1973manipulation,
title = {Manipulation of voting schemes: {A} general result},
author = {Gibbard, Allan},
year = {1973},
journal = {Econometrica},
volume = {41},
number = {4},
pages = {587--601},
url = {https://EconPapers.repec.org/RePEc:ecm:emetrp:v:41:y:1973:i:4:p:587-601}
}

@article{satterthwaite1975strategy,
title = {Strategy-proofness and {A}rrow's conditions: Existence and correspondence theorems for voting procedures and social welfare functions},
journal = {Journal of Economic Theory},
volume = {10},
number = {2},
pages = {187-217},
year = {1975},
issn = {0022-0531},
doi = {https://doi.org/10.1016/0022-0531(75)90050-2},
url = {https://www.sciencedirect.com/science/article/pii/0022053175900502},
author = {Mark Allen Satterthwaite}
}

@article{fishburn1983paradoxes,
 ISSN = {0025570X, 19300980},
 URL = {http://www.jstor.org/stable/2689808},
 author = {Peter C. Fishburn and Steven J. Brams},
 journal = {Mathematics Magazine},
 number = {4},
 pages = {207--214},
 publisher = {Mathematical Association of America},
 title = {Paradoxes of Preferential Voting},
 urldate = {2024-07-18},
 volume = {56},
 year = {1983}
}

@article{MOULIN198853,
title = {Condorcet's principle implies the no show paradox},
journal = {Journal of Economic Theory},
volume = {45},
number = {1},
pages = {53-64},
year = {1988},
issn = {0022-0531},
doi = {https://doi.org/10.1016/0022-0531(88)90253-0},
url = {https://www.sciencedirect.com/science/article/pii/0022053188902530},
author = {Hervé Moulin}
}

@article{barbera1991voting,
 ISSN = {00129682, 14680262},
 URL = {http://www.jstor.org/stable/2938220},
 author = {Salvador Barberà and Hugo Sonnenschein and Lin Zhou},
 journal = {Econometrica},
 number = {3},
 pages = {595--609},
 publisher = {[Wiley, Econometric Society]},
 title = {Voting by committees},
 urldate = {2024-02-21},
 volume = {59},
 year = {1991}
}

@article{yokoo2004effect,
title = {The effect of false-name bids in combinatorial auctions: {N}ew fraud in internet auctions},
journal = {Games and Economic Behavior},
volume = {46},
number = {1},
pages = {174-188},
year = {2004},
issn = {0899-8256},
doi = {https://doi.org/10.1016/S0899-8256(03)00045-9},
url = {https://www.sciencedirect.com/science/article/pii/S0899825603000459},
author = {Makoto Yokoo and Yuko Sakurai and Shigeo Matsubara}
}

@InProceedings{conitzer2008anonymity,
author="Conitzer, Vincent",
editor="Papadimitriou, Christos
and Zhang, Shuzhong",
title="Anonymity-proof voting rules",
booktitle="Internet and Network Economics",
year="2008",
publisher="Springer Berlin Heidelberg",
address="Berlin, Heidelberg",
pages="295--306",
isbn="978-3-540-92185-1"
}

@inproceedings{wagman2008optimal,
  title={Optimal False-Name-Proof Voting Rules with Costly Voting.},
  author={Wagman, Liad and Conitzer, Vincent},
  booktitle={Proceedings of AAAI},
  pages={190--195},
  year={2008}
}

@inproceedings{todo2011false,
author = {Todo, Taiki and Iwasaki, Atsushi and Yokoo, Makoto},
title = {False-name-proof mechanism design without money},
year = {2011},
isbn = {0982657161},
publisher = {International Foundation for Autonomous Agents and Multiagent Systems},
address = {Richland, SC},
booktitle = {The 10th International Conference on Autonomous Agents and Multiagent Systems - Volume 2},
pages = {651–658},
numpages = {8},
location = {Taipei, Taiwan},
series = {AAMAS '11}
}

@article{bu2013unfolding,
title = {Unfolding the mystery of false-name-proofness},
journal = {Economics Letters},
volume = {120},
number = {3},
pages = {559-561},
year = {2013},
issn = {0165-1765},
doi = {https://doi.org/10.1016/j.econlet.2013.06.011},
url = {https://www.sciencedirect.com/science/article/pii/S0165176513002930},
author = {Nanyang Bu}
}

@incollection{todo2020false,
  title={False-name-proof facility location on discrete structures},
  author={Todo, Taiki and Okada, Nodoka and Yokoo, Makoto},
  booktitle={Proceedings of ECAI},
  volume={325},
  pages={227--234},
  year={2020},
  publisher={IOS Press}
}

@article{fioravanti2024false,
  title={False-name-proof and strategy-proof voting rules under separable preferences},
  author={Fioravanti, Federico and Mass{\'o}, Jordi},
  journal={Theory and Decision},
  pages={391--408},
  volume={97},
  year={2024},
  publisher={Springer}
}

@book{moulin1991axioms,
  title={Axioms of cooperative decision making},
  author={Moulin, Herv{\'e}},
  number={},
  year={1988},
  publisher={Cambridge University Press}
}

@article{BONIFACIO2023102845,
title = {Preference restrictions for simple and strategy-proof rules: Local and weakly single-peaked domains},
journal = {Journal of Mathematical Economics},
volume = {106},
pages = {102845},
year = {2023},
issn = {0304-4068},
doi = {https://doi.org/10.1016/j.jmateco.2023.102845},
url = {https://www.sciencedirect.com/science/article/pii/S0304406823000381},
author = {Agustín G. Bonifacio and Jordi Massó and Pablo Neme}
}

@article{ching1998maximal,
  title={A maximal domain for the existence of strategy-proof rules},
  author={Ching, Stephen and Serizawa, Shigehiro},
  journal={Journal of Economic Theory},
  volume={78},
  number={1},
  pages={157--166},
  year={1998},
  publisher={Elsevier}
}

@article{masso2001maximal,
  title={Maximal domain of preferences in the division problem},
  author={Mass{\'o}, Jordi and Neme, Alejandro},
  journal={Games and Economic Behavior},
  volume={37},
  number={2},
  pages={367--387},
  year={2001},
  publisher={Elsevier}
}

@article{serizawa1995power,
  title={Power of voters and domain of preferences where voting by committees is strategy-proof},
  author={Serizawa, Shigehiro},
  journal={Journal of Economic Theory},
  volume={67},
  number={2},
  pages={599--608},
  year={1995},
  publisher={Elsevier}
}

@article{barbera2011strategyproof,
  title={Strategyproof social choice},
  author={Barber{\`a}, Salvador},
  journal={Handbook of Social Choice and Welfare},
  volume={2},
  pages={731--831},
  year={2011},
  publisher={Elsevier}
}

\end{document}